%% file: main.tex
\documentclass[conference,10pt]{IEEEtran}
\usepackage{amssymb,amsmath,amsfonts,bm,epsfig,graphicx,theorem,latexsym}
\usepackage{rotating,setspace,latexsym,epsf,color,subcaption,caption}
\usepackage{cite,comment,url}
\usepackage{ifthen}

\usepackage{tikz}

\newtheorem{Lemma}{Lemma}

\newcommand{\egemen}[1]{\textcolor{blue}{Egemen: #1}}

\input{./math_def}

\newboolean{arXiv}
\setboolean{arXiv}{true}

\begin{document}
\IEEEoverridecommandlockouts
\pagestyle{empty}

\title{Age of Information Optimization in Distributed Sensor Networks with Half-Duplex Channels}

\author{\IEEEauthorblockN{Peng Zou$^1$, Ali Maatouk$^2$, Egemen Erbayat$^3$, Suresh Subramaniam$^3$}
\IEEEauthorblockA{
\textit{$^{1}$Nanjing University of Information Science and Technology}, \textit{$^{2}$Yale University}, \textit{$^{3}$The George Washington University}\\
$^1$\textit{003967@nuist.edu.cn}, $^2$\textit{ali.maatouk@yale.edu}, $^3$\textit{\{erbayat, suresh\}@gwu.edu}} 
}

\maketitle 

\begin{abstract}
Motivated by cooperative distributed networks in which users dynamically alternate between transmit and receive modes under half-duplex constraints, this paper studies the Age of Information (AoI) in a distributed multi-user network using an ALOHA-based protocol. We derive closed-form expressions for the average AoI and formulate an optimization problem over transmission probabilities. After proving the convexity of the problem, we leverage the derived optimality conditions to characterize optimal policies for general network graphs, obtain closed-form solutions for $d$-regular topologies, and derive tractable optimality conditions for star topologies. 
Numerical results confirm that the proposed mechanism can effectively and adaptively determine user-specific optimal transmission probabilities across varying network topologies. These findings contribute to the design of adaptive and efficient distributed networks with enhanced information freshness.

\end{abstract}
\input{chapters/intro}
\input{chapters/system}
\input{chapters/Analysis}
\input{chapters/num}
\input{chapters/concl}

\ifthenelse{\boolean{arXiv}}{}{\newpage}

	\bibliographystyle{ieeetr}  
	\bibliography{references}

\end{document}

%% file: math_def.tex
\def \n2{{N_0 \over 2}}

\def \h5{\hspace{0.5in}}

%% file: chapters/intro.tex
\section{Introduction}

Timely delivery of status updates is a fundamental requirement in many mission-critical communication and control systems. To quantify the timeliness of received information, the Age of Information (AoI) metric was introduced in \cite{Kaul2012} to measure the freshness of status updates at the receiver. In recent years, AoI has become a widely adopted performance metric for systems that demand timely information and has motivated extensive research on optimizing AoI across a variety of application areas \cite{tripathi2022information,zhao2024age,mobihoc2024age,Jones2026}.

  In parallel with these developments, many emerging communication and control systems operate in distributed networked settings, where information generation, transmission, and decision-making are decentralized \cite{he2022collaborative, feng2022joint}. In such systems, multiple agents generate local status updates and exchange information with neighboring nodes to support decision-making and coordination. Instead of a global controller or full system knowledge, each agent operates based on partial and locally available information, which underscores the importance on timely information exchange \cite{tahir2024collaborative}.

Vehicular networks exemplify this class of distributed systems \cite{chen2024end}. Although modern vehicles are equipped with rich onboard sensors, local perception alone is often insufficient due to occlusions, limited sensing range, and rapidly changing traffic conditions \cite{han2023collaborative,malik2023collaborative}. To address these limitations, vehicles exchange locally generated state information such as object detections, trajectories, and motion intent through direct vehicle-to-vehicle communication, which supports cooperative perception and situational awareness \cite{darbha2018benefits,ngo2023cooperative,arnold2020cooperative}.

In these distributed settings, the freshness of received updates is essential \cite{chiariotti2025voi,prasad2021decentralized}. Outdated information about neighboring vehicles or surrounding processes can lead to unsafe or inefficient decisions \cite{zhou2023age}. This makes AoI a natural performance metric for communication protocol design in distributed networks, where the primary objective is timely exchange of locally generated state information among neighboring nodes rather than centralized computation or task offloading\cite{baldesi2019keep,zhou2024information}. In addition to local information exchange mechanisms, gossip networks represent another important communication paradigm for distributed systems, in which nodes broadcast and relay updates, often originating from a common source, to achieve scalable and robust information dissemination under decentralized operation \cite{Kaswan2025, Maranzatto2025, Delfani2025}.

Motivated by this perspective, we consider a distributed network where each user acts both as a sensor/transmitter, generating local updates, and as a receiver, tracking processes observed by neighbors. Users dynamically alternate between transmission and reception under half-duplex constraints, coordinated by an ALOHA-based random access protocol over shared wireless channels, compared with previous AoI-ALOHA works \cite{Wang2023,Zhao2024,Yavascan2021,Munari2021}, where a central monitor receives packets from different sources, our work focuses on a decentralized system. A link succeeds only when one user transmits while its neighbors listen; simultaneous transmissions cause collisions. Unlike classical gossip networks, each user here maintains fresh information about multiple independent local processes, a formulation that naturally reflects cooperative distributed networks, where users continuously track the dynamic states of nearby agents through direct, distributed information exchange. To that end, our contributions are summarized below:

\begin{itemize}
    \item We first introduce the system model and clarify the distinct roles that different users play in the distributed network. Moreover, we design an ALOHA-based transmission protocol to support users in dynamically switching between transmission and reception modes according to network conditions and transmitting data updates to neighboring nodes.
    \item Next, we derive a closed-form expression for the AoI in this network scenario and further investigate the optimization problem based on the probabilities for those users to be in transmit mode. 
    \item Our analysis shows that network topology plays a decisive role in shaping optimal transmission strategies. In $d$-regular graphs, symmetry yields a closed-form optimal transmission probability that depends solely on the node's degree, while star topologies induce asymmetric optimal behavior between hub and leaf nodes. These insights link local connectivity to optimal mode-switching decisions and help interpret solutions in more general networks.
    \item Finally, numerical results confirm that the proposed mechanism can effectively and adaptively determine user-specific optimal probabilities for the users to be transmitters across varying network topologies. These findings contribute to the design of adaptive and efficient distributed networks with enhanced information freshness.
\end{itemize}

%% file: chapters/system.tex
\begin{figure}[!t] 	\centering\includegraphics[width=2.8in]{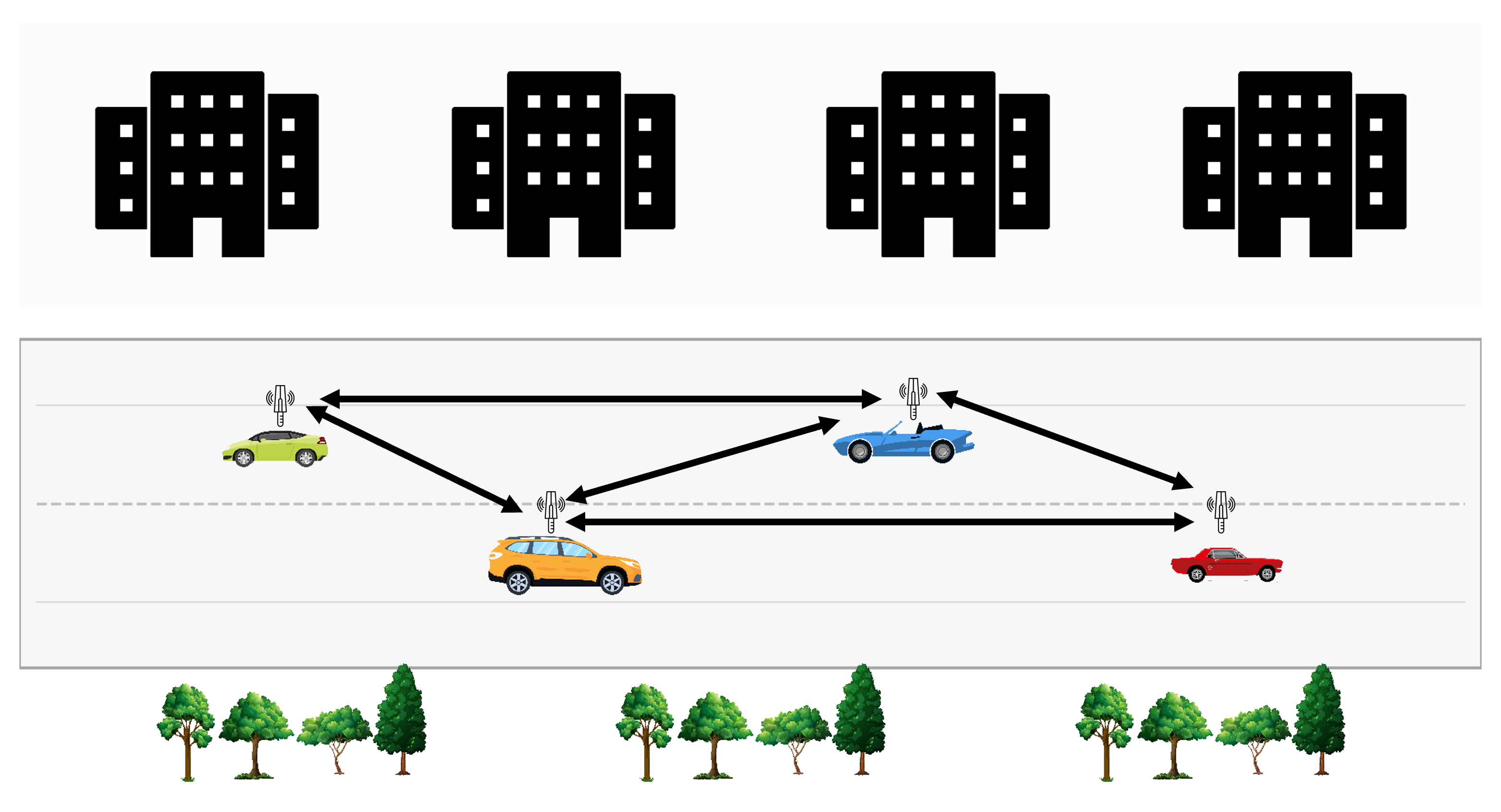}
	\caption{An illustration of data sharing structure.} \label{system_model}\vspace{-0.3in}
\end{figure}
\vspace{-0.05in}
\section{System Model}\vspace{-0.05in}
\label{sec:sys}
In the distributed network shown in Fig. \ref{system_model}, which consists of multiple vehicle users, each user can serve as both a sensor that monitors the environment and a receiving server that processes the sensed data. Since the sensed information of one user may be useful to others, users are allowed to exchange data directly. 

To represent the resulting communication architecture, we model the system as a network of $N$ users interconnected by a fixed bidirectional graph, where each edge indicates a potential communication link between two nodes in either direction. Let $\mathbf{B}_i$ denote the neighbor set of user $i$ that a transmission from user $i$ is received by all users in $\mathbf{B}_i$ if they are in receive mode.  
Building on this model, we now define the performance metric used to evaluate the timeliness of the information exchanged among the users.\vspace{-0.05in}


\subsection{Age of Information}\vspace{-0.05in}
Each user tracks the AoI for updates received from its neighbors. In particular, for any user $j \in \mathbf{B}_i$, we define the AoI of updates generated by user $j$ and observed at user $i$ at time $t$ as\vspace{-0.1in}
\begin{align}
    \Delta_{i,j}(t)=t-G_{i,j}(t),
\end{align}
where $G_{i,j}(t)$ denotes the generation time of the most recently received update from user $j$ at user $i$ by time $t$. Therefore, the average AoI from user $j$ at user $i$ is given by \vspace{-0.1in}
\begin{align}
     \Bar{\Delta}_{i,j}&=\lim_{T \to \infty}\frac{1}{T}\sum_{t=1}^{T}\Delta_{i,j}(t).
\end{align}

We next describe the medium access and transmission mechanism that governs how updates are generated and delivered over the network. 
\subsection{ALOHA-based Transmission Protocol}
We adopt a time-slotted ALOHA protocol with the following assumptions:
\begin{itemize}
    \item At each time slot, user $i$ generates a new update packet with probability $p_i$.
    \item There is no buffer to store packets, and the packets that were not transmitted are dropped.
    \item 
    A user who either has no fresh packet or chooses not to transmit will remain in receive mode. 
    \item User $i$ holding a fresh update packet chooses to enter transmit mode with probability $q_i$ and broadcasts the packet to its neighbors. With probability $1-q_i$, it discards the packet and remains in receive mode instead. 
    \item 
    If two neighboring users transmit simultaneously, the mutual broadcast between them causes a collision, preventing successful packet exchange in either direction. 
    \item If multiple packets are transmitted to node $i$ by its neighbors, then none of these packets are received successfully.
\end{itemize}
\vspace{-0.05in}

\ifthenelse{\boolean{arXiv}}{
\begin{figure}[!t] 	\centering\includegraphics[height=2in]{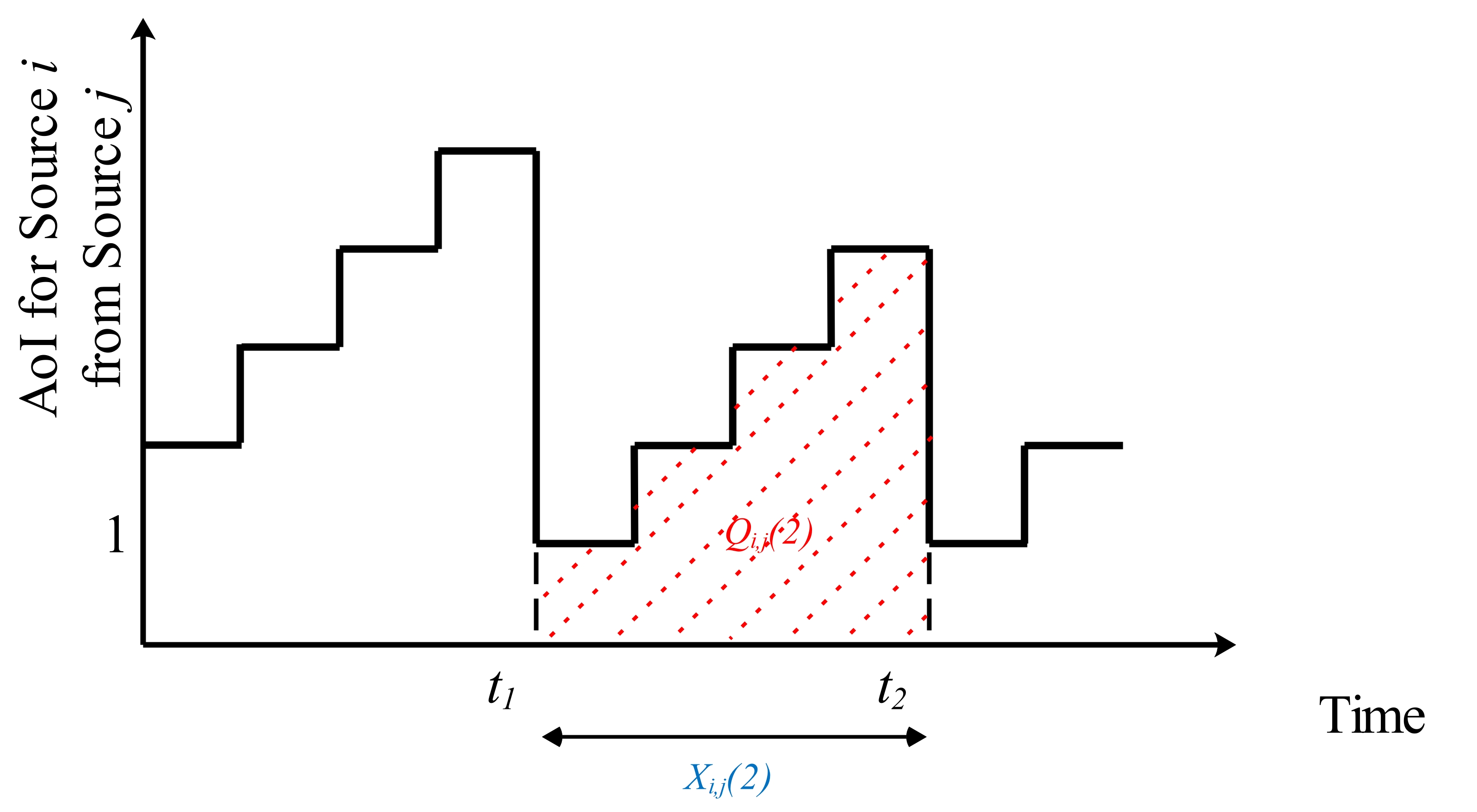}
	\caption{Evolution of AoI for source $i$ in the multiple sources system.} \label{evo}
\end{figure}
We show the AoI evolution for an arbitrary source $i$ receiving packets from source $j$ in Fig.~\ref{evo}.
Under the above transmission model, successful packet deliveries from user $j$ to user $i$ occur randomly over time. For a given network topology, the inter-arrival time between two successfully delivered packets, denoted by $X_{i,j}$, follows a geometric distribution with success probability $\mu_{i,j}$ (to be determined). In Fig.~\ref{evo}, the first packet for source $i$ is transmitted at time slot $t_1$ and the second packet at time slot $t_2$; therefore $X_{i,j}(2)=t_2-t_1$. Let us define $Q_{i,j}(n)$ as the area under the AoI curve between $t_{n-1}$ and $t_n$. For example, $Q_{i,j}(2)$ is shown as the shaded area in Fig.~\ref{evo}.  

Now, $Q_{i,j}(n)=\frac{X_{i,j}(n)(X_{i,j}(n)+1)}{2}$. The average AoI for source $i$ can then be written as:
\begin{align*}
    \Bar{\Delta}_{i,j}=\lim_{t\to\infty}\frac{\sum_{n=1}^{n^*_{i,j}(t)}Q_{i,j}(n)}{\sum_{n=1}^{n^*_{i,j}(t)}X_{i,j}(n)}=\lim_{t\to\infty}\frac{\sum_{n=1}^{n^*_{i,j}(t)}Q_{i,j}(n)/n^*_{i,j}(t)}{\sum_{n=1}^{n^*_{i,j}(t)}X_{i,j}(n)/n^*_{i,j}(t)}
\end{align*}
where $n^*_{i,j}(t)$ is the index of the most recent delivered packet from source $j$ at source $i$ by time slot $t$. Therefore, we have:
\begin{align}\label{avg}
\nonumber\Bar{\Delta}_{i,j}&=\frac{E[Q_{i,j}]}{E[X_{i,j}]}=\frac{\frac{E[X_{i,j}(X_{i,j}+1)]}{2}}{E[X_{i,j}]}\\
&=\frac{E[X_{i,j}^2]}{2E[X_{i,j}]}+\frac{1}{2}=\frac{2-\mu_{i,j}}{2\mu_{i,j}}+\frac{1}{2}=\frac{1}{\mu_{i,j}}.
\end{align}}{Under the above transmission model, successful packet deliveries from user $j$ to user $i$ occur randomly over time. For a given network topology, the inter-arrival 
time between two successfully delivered packets, denoted by $X_{i,j}$, follows a geometric distribution with success probability $\mu_{i,j}$ (to be determined). Then, the average AoI at user $i$ for updates from user $j$ becomes\vspace{-0.1in}
\begin{align}
\Bar{\Delta}_{i,j}&=\frac{\frac{E[X_{i,j}(X_{i,j}+1)]}{2}}{E[X_{i,j}]}=\frac{E[X_{i,j}^2]}{2E[X_{i,j}]}+\frac{1}{2}\\
\nonumber &=\frac{2-\mu_{i,j}}{2\mu_{i,j}}+\frac{1}{2}=\frac{1}{\mu_{i,j}}.
\end{align}
A detailed analysis can be found in our arXiv version\cite{}.}

 Building on this system model, we next study how to optimally choose the transmission probabilities to minimize the overall AoI performance of the network.
\vspace{-0.05in}

%% file: chapters/Analysis.tex
\section{AoI Optimization Analysis}
\label{sec:Analysis}
\vspace{-0.05in}
This section derives a closed-form expression for the average AoI and studies its optimization properties.
\vspace{-0.05in}
\subsection{Average AoI}
\vspace{-0.05in}

For a given network topology, a transmission from user $j$ to user $i$ is successful if user $i$ operates in receive mode and all neighbors of $i$, except $j$, are also receivers. Under this condition, the success probability of link $(i,j)$ is\vspace{-0.05in}
\begin{align}
\mu_{i,j}
= p_j q_j (1 - p_i q_i)
\prod_{k \in \mathcal{B}_i,\, k \neq j} (1 - p_k q_k).
\end{align}
The corresponding average AoI is therefore\vspace{-0.05in}
\begin{align}
\bar{\Delta}_{i,j}
= \frac{1}{\mu_{i,j}}
= \frac{1}{p_j q_j (1 - p_i q_i)
\prod_{k \in \mathcal{B}_i,\, k \neq j} (1 - p_k q_k)}.
\end{align}

To reduce complexity and focus on the impact of transmission decisions, we assume homogeneous update generation probabilities and set $p_i = p$ for all nodes. It is worth noting that our analysis is also applicable to scenarios with asymmetric $p$ values. Let $f(\mathbf{q})$ denote the total average AoI across all directed links. The AoI minimization problem can then be written as\vspace{-0.05in}
\begin{align}
\label{optP}
\min_{\{q_i\}} \quad
f(\mathbf{q})
= \sum_{i=1}^{N} \sum_{j \in \mathcal{B}_i} \bar{\Delta}_{i,j}
\quad
\text{subject to}
\quad
0 \le q_i \le 1, \ \forall i.
\end{align}

 The optimization of this problem is challenging because the decision variable $q$ appears in the denominator of the objective function. To address this, we present the following convexity proof.


\begin{Lemma}
\label{Lem1}
The objective function $f(\mathbf{q})$ defined in \eqref{optP} is convex over the feasible set.
\end{Lemma}

\begin{IEEEproof}
Define
\begin{align}
g_{i,j}(\mathbf{q})
= \frac{1}{p q_j (1 - p q_i)
\prod_{k \in \mathcal{B}_i,\, k \neq j} (1 - p q_k)}.
\end{align}
Since $f(\mathbf{q})$ is a sum of terms $g_{i,j}(\mathbf{q})$, it suffices to show that each $g_{i,j}(\cdot)$ is convex.
Direct analysis of the Hessian of $g_{i,j}$ is analytically intractable. Instead, consider the logarithmic transformation:
\begin{align*}
    G_{i,j}(\mathbf{q})&=\log g_{i,j}(\mathbf{q})\\&=-\log pq_j-\log(1-pq_i)-\sum_{k\in \mathbf{B}_i,k\neq j}\log (1-pq_k). 
\end{align*}
The first derivatives are
\begin{align*}
\frac{\partial G_{i,j}(\mathbf{q})}{\partial q_\ell}
=
\begin{cases}
-\dfrac{1}{q_\ell}, & \ell = j, \\[6pt]
\dfrac{p}{1 - p q_\ell}, & \ell \neq j \text{ and } \ell \in \{i\} \cup \mathcal{B}_i, \\[6pt]
0, & \text{otherwise},
\end{cases}
\end{align*}
and the second derivatives are
\begin{align*}
&\frac{\partial^2 G_{i,j}(\mathbf{q})}{\partial q_\ell^2}
=
\begin{cases}
\dfrac{1}{q_\ell^2}, & \ell = j, \\[6pt]
\dfrac{p^2}{(1 - p q_\ell)^2}, & \ell \neq j \text{ and } \ell \in \{i\} \cup \mathcal{B}_i,
\end{cases}
\\
&\frac{\partial^2 G_{i,j}(\mathbf{q})}{\partial q_\ell \partial q_m} = 0 \ \text{for } \ell \neq m.
\end{align*}

The Hessian of $G_{i,j}(\mathbf{q})$ is diagonal with strictly positive entries, which implies convexity of $G_{i,j}(\mathbf{q})$. Since $G_{i,j}(\mathbf{q})=\log g_{i,j}(\mathbf{q})$ and the exponential preserves convexity for positive functions \cite{boyd2004convex}, $g_{i,j}(\mathbf{q})$ is convex. As a result, $f(\mathbf{q})$ is convex as a sum of convex functions.
\end{IEEEproof}

Lemma~\ref{Lem1} ensures that the AoI minimization problem is a convex problem. This property allows us to apply first-order optimality conditions to characterize optimal transmission probabilities.

\subsection{Optimality Conditions}

Given the convexity of problem \eqref{optP}, the Karush-Kuhn-Tucker conditions are necessary and sufficient for optimality. We focus on interior solutions and derive an explicit expression that must be satisfied by any optimal transmission probability. To that end, we introduce Lagrange multipliers $\lambda_i \ge 0$ and $\nu_i \ge 0$ associated with the constraints $q_i \ge 0$ and $q_i \le 1$, respectively. The Lagrangian of problem \eqref{optP} is
\begin{align}
\mathcal{L}
= f(\mathbf{q}) - \sum_i \lambda_i q_i + \sum_i \nu_i (q_i - 1).
\end{align}

We now derive the stationary condition for an interior solution in order to obtain an explicit characterization of the optimal transmission probability at each node as a function of its local structure as summarized in the following lemma.

\begin{Lemma}
\label{lem:interior_qstar}
If $0 < q_\ell^* < 1$ is optimal for problem \eqref{optP}, then it satisfies
\begin{align}
\label{eq:qstar_closed_form}
q_\ell^*
= \frac{A_\ell}{p(A_\ell + B_\ell)}.
\end{align}
\end{Lemma}

\begin{IEEEproof}
For an interior solution satisfying $0 < q_\ell < 1$, complementary slackness implies $\lambda_\ell = \nu_\ell = 0$, and the stationarity condition reduces to
\begin{align}
\frac{\partial f}{\partial q_\ell} = 0.
\end{align}
Writing $f(\mathbf{q}) = \sum_{i} \sum_{j \in \mathcal{B}_i} g_{i,j}(\mathbf{q})$ and using
\begin{align}
\frac{\partial g_{i,j}}{\partial q_\ell}
= g_{i,j} \frac{\partial \log g_{i,j}}{\partial q_\ell},
\end{align}
the derivative of $f$ with respect to $q_\ell$ can be expressed as
\begin{align}
\label{KKT1}
\frac{\partial f}{\partial q_\ell}
= -\frac{A_\ell}{q_\ell}
+ \frac{p}{1 - p q_\ell} B_\ell,
\end{align}
where
\begin{align}
A_\ell &= \sum_{i : \ell \in \mathcal{B}_i} g_{i,\ell}, \\
B_\ell &= \sum_{j \in \mathcal{B}_\ell} g_{\ell,j}
+ \sum_{i : \ell \in \mathcal{B}_i} \sum_{j \in \mathcal{B}_i,\, j \neq \ell} g_{i,j}.
\end{align}
Setting \eqref{KKT1} to zero yields
\begin{align}
A_\ell (1 - p q_\ell) = p B_\ell q_\ell,
\end{align}
which leads to the interior-point solution
\begin{align}
q_\ell^*
= \frac{A_\ell}{p(A_\ell + B_\ell)}.
\end{align}
\end{IEEEproof}
Lemma~\ref{lem:interior_qstar} provides an explicit interior-point characterization that depends only on the aggregate terms $A_\ell$ and $B_\ell$. This form becomes especially useful in symmetric topologies, where these aggregates simplify and lead to closed-form network-wide solutions.

\subsection{Optimal Probabilities for d-Regular Graphs}

The fixed-point characterization derived in the previous subsection applies to general network topologies but does not provide a closed-form solution easily in most cases. To gain analytical insight, we now focus on symmetric networks in which all nodes experience identical transmission conditions. This symmetry allows the optimality conditions to collapse to a scalar equation.

We consider a $d$-regular graph, where each node has exactly $d$ neighbors. All nodes are assumed to employ a common transmission probability $q$. Under this assumption, all directed links are statistically identical. As a consequence of this symmetry, the AoI minimization problem reduces to a scalar optimization, which results in Lemma~\ref{lem:d-regular}.

\begin{Lemma}
\label{lem:d-regular}
For a $d$-regular graph with homogeneous update probability $p$, the optimal transmission probability that minimizes the total average AoI is
\begin{align}\label{optq_s}
q^* = \min \left\{ \frac{1}{p(d+1)},\, 1 \right\}.
\end{align}
\end{Lemma}

\begin{IEEEproof}
Under the symmetry assumptions, the success probability of any directed link is
\begin{align}
\mu = p q (1 - p q)^d.
\end{align}
Since all nodes contribute equally to the total average AoI, minimizing the total AoI is equivalent to minimizing $1 / \mu$. Then, taking the derivative of $1 / \mu$ with respect to $q$ and setting it equal to zero yields the first-order condition

\begin{align}
q = \frac{1}{p(d+1)}.
\end{align}
Afterward, imposing the feasibility constraint $0 \le q \le 1$ completes the proof.
\end{IEEEproof}

Lemma~\ref{lem:d-regular} shows that the optimal transmission probability decreases as the node degree increases, which reflects the higher level of collisions present in denser networks.

\subsection{Optimal Probabilities for the Star Topology}

We next consider the star topology, which features a highly asymmetric connectivity structure. One node acts as a central hub connected to all other nodes, while the remaining $N-1$ nodes have no mutual connections. Owing to this asymmetry, the star topology requires a distinct analytical approach. From our analysis of the $d$-regular graph, we can observe that symmetric nodes should share the same optimal $q$. Therefore, we let $q_1$ denote the probability that the central node is a transmitter and $q_2$ that of the leaf node. For analytical clarity, and without loss of generality, we assume the update generation probability $p=1$ for all nodes. The AoI minimization problem for this topology reduces to the objective function
\begin{align}
f(q_1, q_2) = \frac{N-1}{q_2 (1-q_1) (1-q_2)^{\,N-2}} + \frac{N-1}{q_1 (1-q_2)},
\end{align}
where $ N \geq 2$. 
Next, we introduce the following lemma.
\begin{Lemma}
    The optimal transmission probability for the star topology can be derived by solving the following equations
    \begin{align}
1 - (N-1)q_2 = q_2^{3/2} (1-q_2)^{\frac{N-3}{2}}. \label{starq2}\\
q_1^2 = q_2 (1-q_1)^2 (1-q_2)^{N-3}. \label{q1andq2}
\end{align}
\end{Lemma}
\begin{IEEEproof}
The partial derivative with respect to $q_1$ is
\begin{align*}
\frac{\partial f}{\partial q_1} = \frac{N-1}{q_2 (1-q_2)^{N-2}}  \frac{1}{(1-q_1)^2} - \frac{N-1}{q_1^2 (1-q_2)}.
\end{align*}
Setting $\partial f / \partial q_1 = 0$ leads to
\begin{align*}
\frac{1}{q_2 (1-q_2)^{N-2}}  \frac{1}{(1-q_1)^2} = \frac{1}{q_1^2 (1-q_2)}.
\end{align*}
Rearranging terms yields the first optimality condition
\begin{align*}
q_1^2 = q_2 (1-q_1)^2 (1-q_2)^{N-3}. 
\end{align*}

The partial derivative with respect to $q_2$ is more involved
\begin{align*}
\frac{\partial f}{\partial q_2} = & -(N-1)[ \frac{1}{q_2^2 (1-q_1)(1-q_2)^{N-2}}\\
&+ \frac{(2-N)}{q_2 (1-q_1)(1-q_2)^{N-1}} ] + \frac{N-1}{q_1 (1-q_2)^2}.
\end{align*}
Setting $\partial f / \partial q_2 = 0$ yields
\begin{align}
-\left[ \frac{1-q_2}{q_2} + (2-N) \right] + q_2 (1-q_2)^{N-3}\frac{1-q_1}{q_1} = 0.
\end{align}
Substituting the relation for $1/(1-q_1)$ derived from rearranging (\ref{q1andq2}), namely $q_1/(1-q_1) = \sqrt{q_2}(1-q_2)^{(N-3)/2}$, into the above equation leads, after algebraic manipulation, to the second optimality condition
\begin{align*}
1 - (N-1)q_2 = q_2^{3/2} (1-q_2)^{\frac{N-3}{2}}. 
\end{align*}
Equation (\ref{starq2}) is a univariate equation in $q_2$ that can be solved numerically for a given $N$. Subsequently, $q_1^*$ can be obtained from (\ref{q1andq2}).
\end{IEEEproof}
For $N=2$, the network consists of a single link, and the AoI minimization problem trivially yields $q_1=q_2=1/2$. For $N=3$, the star topology reduces to a line topology, and we obtain $q_1=q_2=\frac{3-\sqrt5}{2}$.
For $N>3$, the optimal leaf-node transmission probability $q_2$ is determined by the nonlinear condition in~\eqref{starq2}. By squaring both sides of~\eqref{starq2} and expanding $(1-q_2)^{N-3}$ using the binomial theorem, this condition can be equivalently written as the polynomial equation
\begin{align}
\sum_{k=0}^{N-3} \binom{N-3}{k}(-1)^k q_2^{k+3}
+
2(N-1)q_2
-(N-1)^2 q_2^2
=1, 
\end{align}
which gives a unique solution in the interval $(0,1)$. Once $q_2$ is obtained numerically, the optimal central-node probability $q_1^\star$ follows directly from~\eqref{q1andq2}. 

While a closed-form solution remains elusive for more complex topologies, the proven convexity of the problem enables efficient and reliable computation of the global optimum using established convex optimization tools (e.g., CVX\cite{cvx}).


%% file: chapters/num.tex
\section{Numerical Results}
\label{sec:num}
In this section, we present analytical results for different network topologies, all of which are validated through Monte Carlo simulations. For clarity of presentation, the figures display the average AoI per source, computed as $\frac{f(\mathbf{q})}{N}$  shown in the figures. This scaling does not affect our optimization analysis.

\begin{figure}[!t]
    \centering
    \includegraphics[height=1.5in]{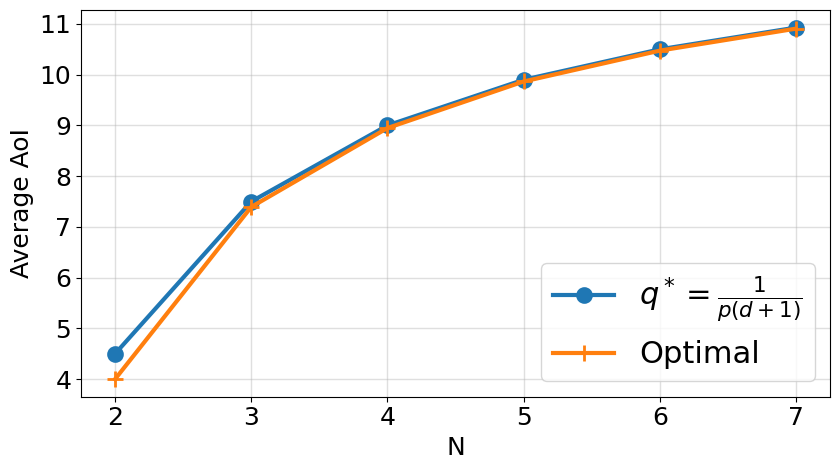}
    \caption{Average AoI with respect to $N$ for line topology with fixed $p=1$. }
    \label{Num1}
     \vspace{0.5cm}
\centering
    \begin{tabular}{|c|c|c|c|c|c|c|c|}
    \hline
    index $i$ & 1 & 2 & 3 & 4 & 5 & 6 & 7 \\
    \hline
    $\mathcal{N}(\mathcal{B}_i)$ & 1 & 2 & 2 & 2 & 2 & 2 & 1 \\
    \hline
    $q_i^*$ & 0.36 & 0.35 & 0.34 & 0.34 & 0.34 & 0.35 & 0.36 \\
    \hline
    \end{tabular}
    \captionof{table}{Optimal $q$ values for different indices for a line topology. }
    \label{Table1} 
\end{figure}
\begin{figure}[!t]
    \centering
    \includegraphics[height=1.5in]{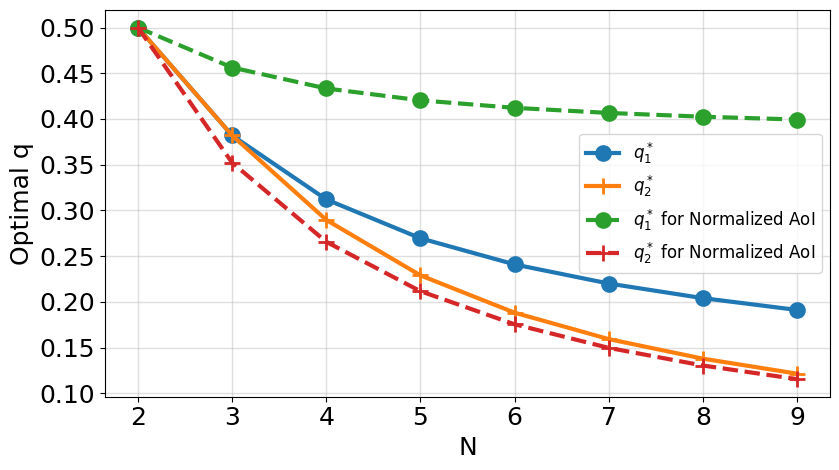}
    \caption{Optimal $q$ with respect to $N$ for a star topology with fixed $p = 1$.}
    \label{Num2}\vspace{-0.2in}
\end{figure}

We start with Figure \ref{Num1}, in which we compare $q^{*} = \frac{1}{p(d+1)}$ from Eq. \eqref{optq_s} with the optimal solution obtained using CVX for a line topology. It can be observed that when $N$ is small, $q^{*} = \frac{1}{p(d+1)}$ still shows a noticeable gap from the true optimal solution derived from CVX. However, as $N$ becomes large, the difference becomes very small. This is because, for large $N$, the line topology can be approximated as a symmetric $d$-regular graph. Additionally, in Table \ref{Table1}, we list the optimal $q$ for each node when $N = 7$, where $\mathcal{N}(\mathcal{B}_i)$ denotes the number of neighbors for user $i$ for a line topology. It can also be seen that the actual optimal $q$ values exhibit a wavelike distribution. Although the optimum for a given node depends only on its immediate neighbors, these neighbors are also influenced by their own neighbors, ultimately causing the optimal solutions for all users to become interdependent.


Next, in Figure \ref{Num2}, we compare optimal $ q_1 $ and $ q_2 $ with increasing $ N $ in a star topology. It can be observed that as $ N $ increases, both $ q_1 $ and $ q_2 $ decrease, but $ q_1 $ decreases at a slower rate, indicating that the transmission probability requirement for the central node is higher than that for the branch nodes. On the other hand, we also examine the case where the objective function is the average AoI normalized by the number of neighbors, where the objective function $f'(\mathbf{q})
= \frac{1}{N}\sum_{i=1}^{N} \frac{1}{\mathcal{N}(\mathcal{B}_i)}\sum_{j \in \mathcal{B}_i} \bar{\Delta}_{i,j}$. Note that this modification does not significantly alter our analysis. In this scenario, it is surprising to note that $ q_2 $ decreases much faster than $ q_1 $ and for the central node, its cumulative AoI is already averaged due to its larger number of neighbors, yet its optimal transmission rate becomes higher compared to the case without averaging.

\begin{figure}[t]
    \centering

    \begin{subfigure}{0.49\linewidth}
        \centering
        \includegraphics[width=\linewidth]{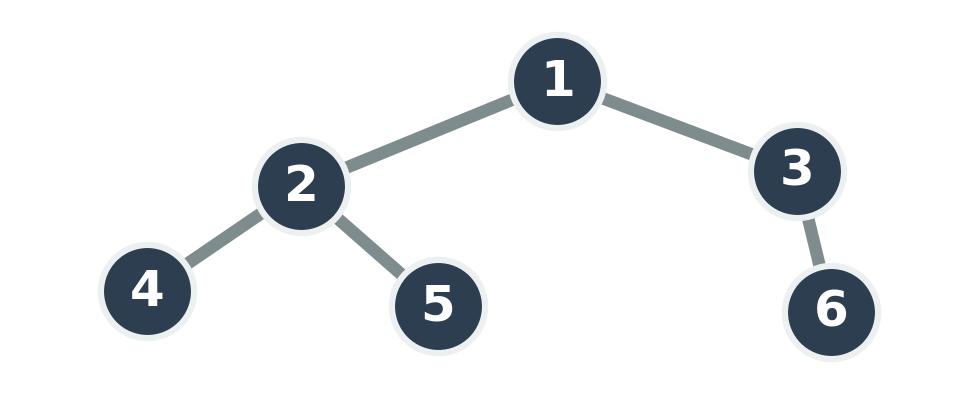}
        \caption{Tree topology.}
        \label{fig:sub1}
    \end{subfigure}\hfill
    \begin{subfigure}{0.49\linewidth}
        \centering
        \includegraphics[width=\linewidth]{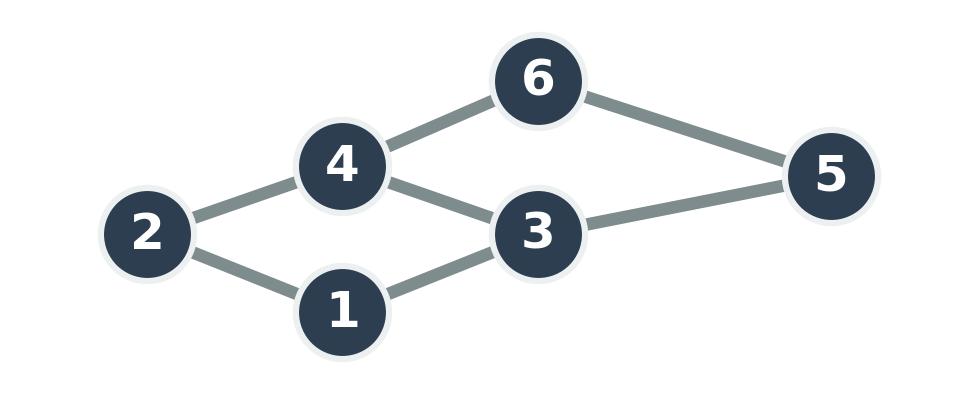}
        \caption{Grid topology.}
        \label{fig:sub2}
    \end{subfigure}

    \vspace{0.15cm}

    \begin{subfigure}{0.49\linewidth}
        \centering
        \includegraphics[width=\linewidth]{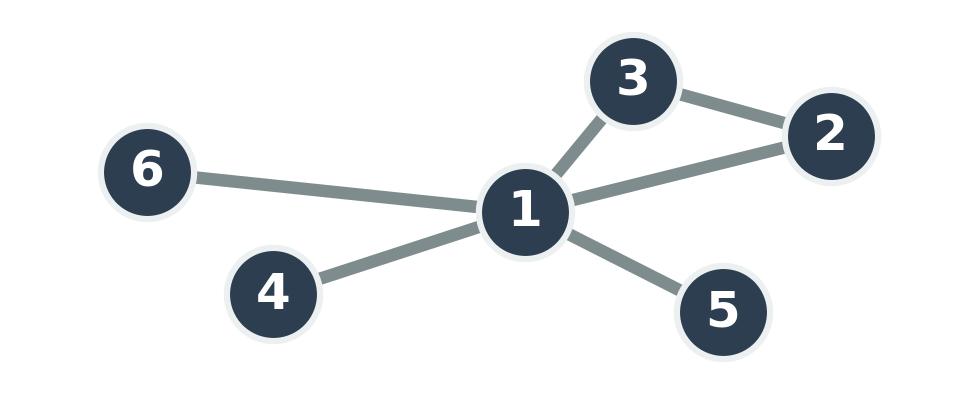}
        \caption{Asymmetric star topology.}
        \label{fig:sub3}
    \end{subfigure}\hfill
    \begin{subfigure}{0.49\linewidth}
        \centering
        \includegraphics[width=\linewidth]{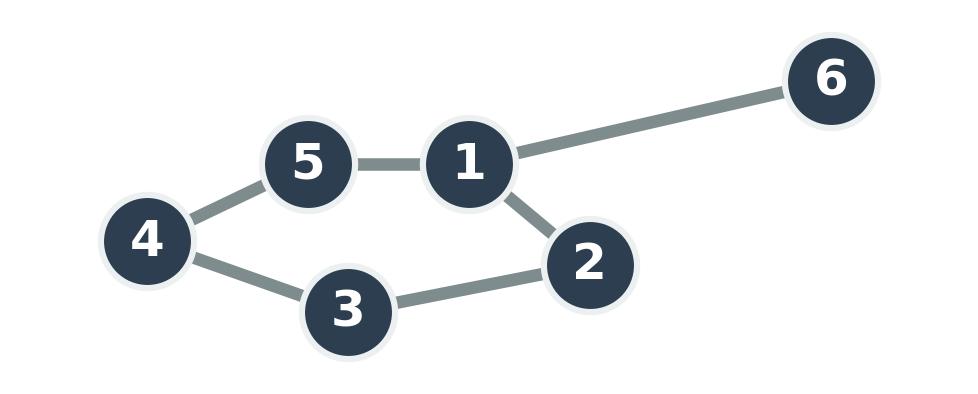}
        \caption{Asymmetric circle topology.}
        \label{fig:sub4}
    \end{subfigure}

    \caption{Four different asymmetric topologies.}
    \vspace{-5pt}
    \label{fig:Num3}
\end{figure}

\begin{figure}[!t] 	\centering\includegraphics[height=2.2in]{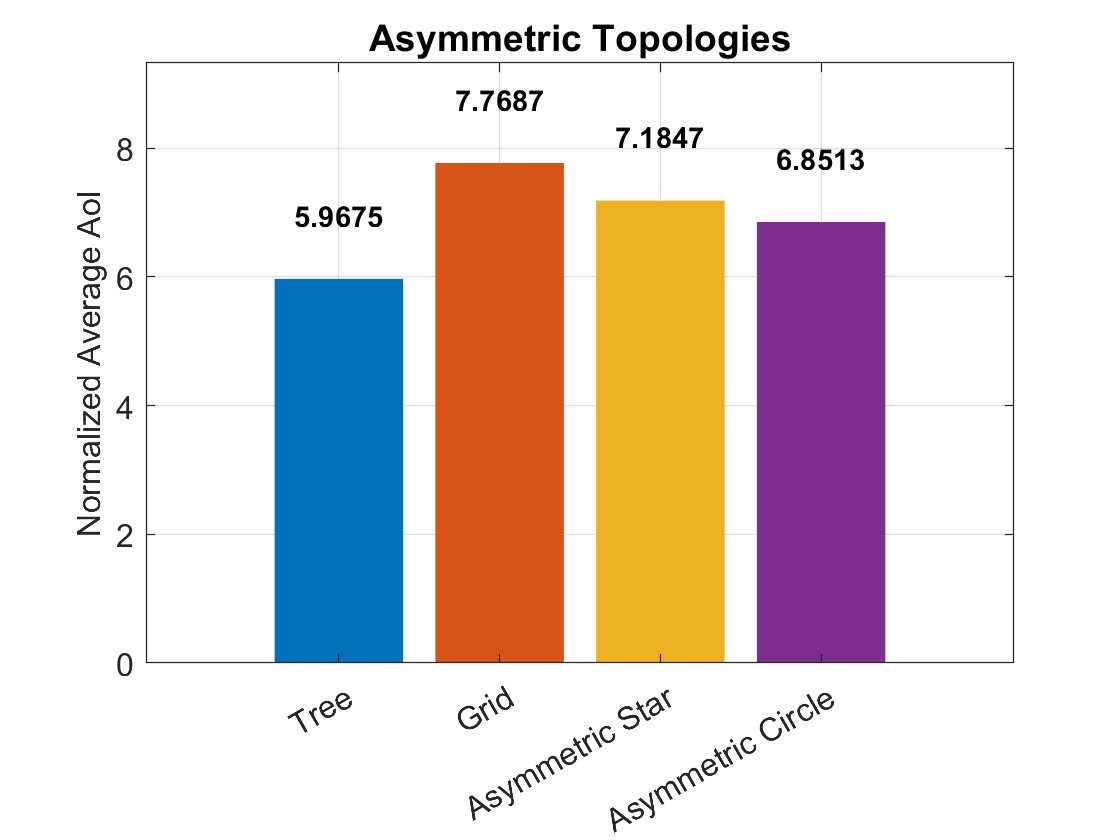}
    \vspace{-5pt}
	\caption{Normalized Average AoI for different asymmetric topologies with fixed $p=1$, $N=6$. } \label{Num3-b}\vspace{-0.2in}
\end{figure}

In Figure \ref{Num3-b}, we investigate the performance of normalized AoI across different asymmetric topologies. We compare tree topology, grid topology, asymmetric star topology, and asymmetric circle topology, as illustrated in Figure \ref{fig:Num3}. Our results demonstrate that the complexity of the network structure significantly affects the normalized AoI. If we assess network complexity by the average number of neighbors per node, we observe that networks with a larger average neighbor count tend to exhibit higher normalized AoI after normalization.

%% file: chapters/concl.tex
\section{Conclusion}
\label{sec:Con}
In this paper, we proposed a cooperative distributed networks framework where users dynamically switch between transmission and reception under half-duplex constraints. We introduced a system model to capture such adaptive interactions and designed an ALOHA-based transmission protocol to coordinate communication. A closed‑form expression for the AoI was derived, and an optimization study was conducted with respect to the probability for users to act as transmitters. By establishing the convexity of the formulated problem, we characterized optimal policies for general network graphs, obtained closed-form solutions for $d$-regular topologies, and derived tractable optimality conditions for star networks. Numerical results demonstrated that the proposed mechanism can effectively and adaptively determine user-specific optimal transmission probabilities across different network topologies.